\documentclass[fleqn]{llncs}
\usepackage{pgatmp}

\title{Program Algebra for Turing-Machine Programs}
\author{J.A. Bergstra and C.A. Middelburg}
\institute{Informatics Institute, Faculty of Science,
           University of Amsterdam, \\
           Science Park~904, 1098~XH Amsterdam, the Netherlands \\
           \email{J.A.Bergstra@uva.nl,C.A.Middelburg@uva.nl}}

\begin{document}

\maketitle

\begin{abstract}
This paper presents an algebraic theory of instruction sequences with
instructions for Turing tapes as basic instructions, the behaviours 
produced by the instruction sequences concerned under execution, and the 
interaction between such behaviours and Turing tapes provided by an 
execution environment.
This theory provides a setting for the development of theory in areas 
such as computability and computational complexity that distinguishes 
itself by offering the possibility of equational reasoning and being 
more general than the setting provided by a known version of the 
Turing-machine model of computation.
The theory is essentially an instantiation of a parameterized 
algebraic theory which is the basis of a line of research in which 
issues relating to a wide variety of subjects from computer science have 
been rigorously investigated thinking in terms of instruction sequences.

\begin{keywords}
program algebra, thread algebra, model of computation, Turing-machine 
program, computability, computational complexity. 
\end{keywords}%
\begin{classcode}
F.1.1, F.1.3, F.4.1.
\end{classcode}
\end{abstract}

\section{Introduction}
\label{sect-intro}

This paper introduces an algebraic theory that provides a setting for 
the development of theory in areas such as computability and 
computational complexity.
The setting in question distinguishes itself by offering the possibility 
of equational reasoning and being more general than the setting provided 
by a known version of the Turing-machine model of computation.
Many known and unknown versions of the Turing-machine model of 
computation can be dealt with by imposing apposite restrictions.
We expect that the generality is conducive to the investigation of novel 
issues in areas such as computability and computational complexity.
This expectation is based on our experience with a comparable algebraic 
theory of instruction sequences, where instructions operate on Boolean 
registers, in previous work 
(see~\cite{BM13b,BM13a,BM14a,BM14e,BM13c,BM18a}).

It is often said that a program is an instruction sequence.
If this characterization has any value, it must be the case that it is 
somehow easier to understand the concept of an instruction sequence than 
to understand the concept of a program. 
In tune with this, the first objective of the work on instruction 
sequences that started with~\cite{BL02a}, and of which an enumeration is 
available at~\cite{SiteIS}, is to understand the concept of a program. 
The notion of an instruction sequence appears in the work in question as 
a mathematical abstraction for which the rationale is based on this 
objective. 

The structure of the mathematical abstraction at issue has been 
determined in advance with the hope of applying it in diverse 
circumstances where in each case the fit may be less than perfect. 
Until now, the work in question has, among other things, yielded an 
approach to non-uniform computational complexity where instruction 
sequence length is used as complexity measure, a contribution to the 
conceptual analysis of the notion of an algorithm, and new insights into 
such diverse issues as the halting problem, program parallelization for 
the purpose of explicit multi-threading and virus detection.

The basis of all the work in question (see~\cite{SiteIS}) is the 
combination of an algebraic theory of single-pass instruction sequences, 
called program algebra, and an algebraic theory of mathematical objects 
that represent the behaviours produced by instruction sequences under 
execution, called basic thread algebra, extended to deal with the 
interaction between such behaviours and components of an execution 
environment.
This combination is parameterized by a set of basic instructions and a 
set of objects that represent the behaviours exhibited by the components 
of an execution environment.

The current paper contains a simplified presentation of the 
instantiation of this combination in which all possible instructions to 
read out or alter the content of the cell of a Turing tape under the 
tape's head and to optionally move the head in either direction by one 
cell are taken as basic instructions and Turing tapes are taken as the 
components of an execution environment.
The rationale for taking certain instructions as basic instructions is 
that the instructions concerned are sufficient to compute each function 
on bit strings.
However, shorter instruction sequences may be possible if certain 
additional instructions are taken as basic instructions.
That is why we opted for the most general instantiation.

An instantiation in which instructions to read out or alter the content 
of a Boolean register are taken as basic instructions and Boolean 
registers are taken as the components of an execution environment 
turned out to be useful to rigorous investigations of issues relating to 
non-uniform computational complexity and algorithm efficiency 
(see e.g.~\cite{BM13a,BM14e}).
We expect that the instantiation presented in this paper can be useful 
to rigorous investigations relating to uniform computational complexity 
and algorithm efficiency.

Program algebra and basic thread algebra were first presented 
in~\cite{BL02a}.%
\footnote
{In that paper and the first subsequent papers, basic thread algebra 
 was introduced under the name basic polarized process algebra.}
The extension of basic thread algebra referred to above, an extension to 
deal with the interaction between the behaviours produced by instruction 
sequences under execution and components of an execution environment, 
was first presented in~\cite{BM09k}.
The presentation of the extension is rather involved because it is
parameterized and owing to this covers a generic set of basic 
instructions and a generic set of execution environment components.
In the current paper, a much less involved presentation is obtained by 
covering only the case where the execution environment components are
Turing tapes and the basic instructions are instructions to read out or 
alter the content of the cell of a Turing tape under the tape's head and 
to optionally move the head in either direction by one cell.

After the presentation in question, we make precise in the setting of 
the presented theory what it means that a given instruction sequence 
computes a given partial function on bit strings, introduce the notion 
of a single-tape Turing-machine program in the setting, give a result 
concerning the computational power of such programs, and give a result 
relating the complexity class $\mathbf{P}$ to the functions that can be 
computed by such programs in polynomial time.
We also give a simple example of a single-tape Turing-machine program.
This example is only given to illustrate the close resemblance of such 
programs to transition functions of Turing machines.
The notation that is used for Turing-machine programs is intended for 
theoretical purposes and not for actual programming.

This paper is organized as follows.
First, we introduce program algebra (Section~\ref{sect-PGA}) and basic 
thread algebra (Section~\ref{sect-BTA}) and extend their combination to 
make precise which behaviours are produced by instruction sequences 
under execution (Section~\ref{sect-TE-BC}).
Next, we present the instantiation of the resulting theory in which all 
possible instructions to read out or alter the content of the cell of a 
Turing tape under the tape's head and to optionally move the head in 
either direction by one cell are taken as basic instructions 
(Section~\ref{sect-PGAtt}), introduce an algebraic theory of Turing-tape 
families (Section~\ref{sect-TTFA}), and extend the combination of the 
theories presented in the two preceding sections to deal with the 
interaction between the behaviours of instruction sequences under 
execution and Turing tapes (Section~\ref{sect-TSI}).
Then, we formalize in the setting of the resulting theory what it means 
that a given instruction sequence computes a given partial function on 
bit strings (Section~\ref{sect-comput-boolstring-fnc}) and give as an 
example an instruction sequence that computes the non-zeroness test 
function (Section~\ref{sect-example}).
Finally, we make some concluding remarks (Section~\ref{sect-concl}).

In this paper, some familiarity with algebraic specification, 
computability, and computational complexity is assumed.
The relevant notions are explained in many handbook chapters and 
textbooks, e.g.~\cite{EM85a,ST12a,Wir90a} for notions concerning 
algebraic specification and~\cite{AB09a,HS11a,Sip13a} for notions
concerning computability and computational complexity.

Sections~\ref{sect-PGA}--\ref{sect-TE-BC} are largely shortened versions 
of Sections~2--4 of~\cite{BM18b}, which, in turn, draw from the 
preliminary sections of several earlier papers.

\section{Program Algebra}
\label{sect-PGA}

In this section, we introduce \PGA\ (ProGram Algebra).
The starting-point of \PGA\ is the perception of a program as a 
single-pass instruction sequence, i.e.\ a possibly infinite sequence of 
instructions of which each instruction is executed at most once and can 
be dropped after it has been executed or jumped over.
The concepts underlying the primitives of program algebra are common in
programming, but the particular form of the primitives is not common.
The predominant concern in the design of \PGA\ has been to achieve 
simple syntax and semantics, while maintaining the expressive power of 
arbitrary finite control.

It is assumed that a fixed but arbitrary set $\BInstr$ of 
\emph{basic instructions} has been given.
$\BInstr$ is the basis for the set of instructions that may occur in 
the instruction sequences considered in \PGA.
The intuition is that the execution of a basic instruction may modify a
state and must produce the Boolean value $\False$ or $\True$ as reply at 
its completion.
The actual reply may be state-dependent.
In applications of \PGA, the instructions taken as basic instructions 
vary from instructions relating to Boolean registers via instructions 
relating to Turing tapes to machine language instructions of actual 
computers.   

The set of instructions of which the instruction sequences considered 
in \PGA\ are composed is the set that consists of the following 
elements:
\begin{itemize}
\item
for each $a \in \BInstr$, a \emph{plain basic instruction} $a$;
\item
for each $a \in \BInstr$, a \emph{positive test instruction} $\ptst{a}$;
\item
for each $a \in \BInstr$, a \emph{negative test instruction} $\ntst{a}$;
\item
for each $l \in \Nat$, a \emph{forward jump instruction} $\fjmp{l}$;
\item
a \emph{termination instruction} $\halt$.
\end{itemize}
We write $\PInstr$ for this set.
The elements from this set are called \emph{primitive instructions}.

Primitive instructions are the elements of the instruction sequences 
considered in \PGA.
On execution of such an instruction sequence, these primitive 
instructions have the following effects:
\begin{itemize}
\item
the effect of a positive test instruction $\ptst{a}$ is that basic
instruction $a$ is executed and execution proceeds with the next
primitive instruction if $\True$ is produced and otherwise the next
primitive instruction is skipped and execution proceeds with the
primitive instruction following the skipped one --- if there is no
primitive instruction to proceed with,
inaction occurs;
\item
the effect of a negative test instruction $\ntst{a}$ is the same as
the effect of $\ptst{a}$, but with the role of the value produced
reversed;
\item
the effect of a plain basic instruction $a$ is the same as the effect
of $\ptst{a}$, but execution always proceeds as if $\True$ 
is produced;
\item
the effect of a forward jump instruction $\fjmp{l}$ is that execution
proceeds with the $l$th next primitive instruction --- if $l$ equals $0$ 
or there is no primitive instruction to proceed with, inaction occurs;
\item
the effect of the termination instruction $\halt$ is that execution 
terminates.
\end{itemize}
Inaction occurs if no more basic instructions are executed, but 
execution does not terminate.

\PGA\ has one sort: the sort $\InSeq$ of \emph{instruction sequences}. 
We make this sort explicit to anticipate the need for many-sortedness
later on.
To build terms of sort $\InSeq$, \PGA\ has the following constants and 
operators:
\begin{itemize}
\item
for each $u \in \PInstr$, 
the \emph{instruction} constant $\const{u}{\InSeq}$\,;
\item
the binary \emph{concatenation} operator 
$\funct{\ph \conc \ph}{\InSeq \x \InSeq}{\InSeq}$\,;
\item
the unary \emph{repetition} operator 
$\funct{\ph\rep}{\InSeq}{\InSeq}$\,.
\end{itemize}
Terms of sort $\InSeq$ are built as usual in the one-sorted case.
We assume that there are infinitely many variables of sort $\InSeq$, 
including $X,Y,Z$.
We use infix notation for concatenation and postfix notation for
repetition.

A \PGA\ term in which the repetition operator does not occur is called 
a \emph{repetition-free} \PGA\ term.

One way of thinking about closed \PGA\ terms is that they represent 
non-empty, possibly infinite sequences of primitive instructions with 
finitely many distinct suffixes.
The instruction sequence represented by a closed term of the form
$t \conc t'$ is the instruction sequence represented by $t$
concatenated with the instruction sequence represented by $t'$.%
\footnote
{The concatenation of an infinite sequence with a finite or infinite 
sequence yields the former sequence.}
The instruction sequence represented by a closed term of the form 
$t\rep$ is the instruction sequence represented by $t$ concatenated 
infinitely many times with itself.
A closed \PGA\ term represents a finite instruction sequence if and 
only if it is a closed repetition-free \PGA\ term.
 
The axioms of \PGA\ are given in Table~\ref{axioms-PGA}.%
\begin{table}[!t]
\caption{Axioms of \PGA} 
\label{axioms-PGA}
\begin{eqntbl}
\begin{axcol}
(X \conc Y) \conc Z = X \conc (Y \conc Z)             & \axiom{PGA1}  \\
(X^n)\rep = X\rep                                     & \axiom{PGA2}  \\
X\rep \conc Y = X\rep                                 & \axiom{PGA3}  \\
(X \conc Y)\rep = X \conc (Y \conc X)\rep             & \axiom{PGA4} 
\eqnsep
\fjmp{k{+}1} \conc u_1 \conc \ldots \conc u_k \conc \fjmp{0} =
\fjmp{0} \conc u_1 \conc \ldots \conc u_k \conc \fjmp{0} 
                                                      & \axiom{PGA5}  \\
\fjmp{k{+}1} \conc u_1 \conc \ldots \conc u_k \conc \fjmp{l} =
\fjmp{l{+}k{+}1} \conc u_1 \conc \ldots \conc u_k \conc \fjmp{l}
                                                      & \axiom{PGA6}  \\
(\fjmp{l{+}k{+}1} \conc u_1 \conc \ldots \conc u_k)\rep =
(\fjmp{l} \conc u_1 \conc \ldots \conc u_k)\rep       & \axiom{PGA7}  \\
\fjmp{l{+}k{+}k'{+}2} \conc u_1 \conc \ldots \conc u_k \conc
(v_1 \conc \ldots \conc v_{k'{+}1})\rep = {} \\ \phantom{{}{+}k'}
\fjmp{l{+}k{+}1} \conc u_1 \conc \ldots \conc u_k \conc
(v_1 \conc \ldots \conc v_{k'{+}1})\rep               & \axiom{PGA8} 
\end{axcol}
\end{eqntbl}
\end{table}
In this table, 
$u$, $u_1,\ldots,u_k$ and $v_1,\ldots,v_{k'+1}$ stand for arbitrary 
primitive instructions from $\PInstr$, 
$k$, $k'$, and $l$ stand for arbitrary natural numbers from $\Nat$, and
$n$ stands for an arbitrary natural number from $\Natpos$.%
\footnote
{We write $\Natpos$ for the set $\set{n \in \Nat \where n \geq 1}$ of
positive natural numbers.}
For each $n \in \Natpos$, the term $t^n$, where $t$ is a \PGA\ term, 
is defined by induction on $n$ as follows: $t^1 = t$, and 
$t^{n+1} = t \conc t^n$.

Let $t$ and $t'$ be closed \PGA\ terms.
Then $t = t'$ is derivable from the axioms of \PGA\ iff $t$ and $t'$ 
represent the same instruction sequence after changing all chained jumps 
into single jumps and making all jumps as short as possible. 
Moreover, $t = t'$ is derivable from PGA1--PGA4 iff $t$ and $t'$ 
represent the same instruction sequence. 
We write \PGAisc\ for the algebraic theory whose sorts, constants and
operators are those of \PGA, but whose axioms are PGA1--PGA4.

The informal explanation of closed \PGA\ terms as sequences of primitive 
instructions given above can be looked upon as a sketch of the intended 
model of the axioms of \PGAisc.
This model, which is described in detail in, for example, \cite{BM12b}, 
is an initial model of the axioms of \PGAisc.
Henceforth, the instruction sequences of the kind considered in \PGA\ 
are called \PGA\ instruction sequences.

\section{Basic Thread Algebra for Finite and Infinite Threads}
\label{sect-BTA}

In this section, we introduce \BTA\ (Basic Thread Algebra) and an 
extension of \BTA\ that reflects the idea that infinite threads are 
identical if their approximations up to any finite depth are identical.

\BTA\ is concerned with mathematical objects that model in a direct 
way the behaviours produced by \PGA\ instruction sequences under 
execution.
The objects in question are called threads.
A thread models a behaviour that consists of performing basic actions in 
a sequential fashion.
Upon performing a basic action, a reply from an execution environment
determines how the behaviour proceeds subsequently.
The possible replies are the Boolean values $\False$ and $\True$.

The basic instructions from $\BInstr$ are taken as basic actions.
Besides, $\Tau$ is taken as a special basic action.
It is assumed that $\Tau \notin \BAct$.
We write $\BActTau$ for $\BAct \union \set{\Tau}$.

\BTA\ has one sort: the sort $\Thr$ of \emph{threads}. 
We make this sort explicit to anticipate the need for many-sortedness
later on.
To build terms of sort $\Thr$, \BTA\ has the following constants and 
operators:
\begin{itemize}
\item
the \emph{inaction} constant $\const{\DeadEnd}{\Thr}$;
\item
the \emph{termination} constant $\const{\Stop}{\Thr}$;
\item
for each $\alpha \in \BActTau$, the binary 
\emph{postconditional composition} operator 
$\funct{\pcc{\ph}{\alpha}{\ph}}{\Thr \x \Thr}{\Thr}$.
\end{itemize}
Terms of sort $\Thr$ are built as usual in the one-sorted case. 
We assume that there are infinitely many variables of sort $\Thr$, 
including $x,y,z$.
We use infix notation for postconditional composition.
We introduce \emph{basic action prefixing} as an abbreviation: 
$\alpha \bapf t$, where $\alpha \in \BActTau$ and $t$ is a \BTA\ term, 
abbreviates $\pcc{t}{\alpha}{t}$.
We treat an expression of the form $\alpha \bapf t$ and the \BTA\ term 
that it abbreviates as syntactically the same.

Closed \BTA\ terms are considered to represent threads.
The thread represented by a closed term of the form 
$\pcc{t}{\alpha}{t'}$ models the behaviour that first performs $\alpha$, 
and then proceeds as the behaviour modeled by the thread represented by 
$t$ if the reply from the execution environment is $\True$ and proceeds 
as the behaviour modeled by the thread represented by $t'$ if the reply 
from the execution environment is $\False$. 
Performing $\Tau$, which is considered performing an internal action,
always leads to the reply $\True$.
The thread represented by $\Stop$ models the behaviour that does nothing 
else but terminate and the thread represented by $\DeadEnd$ models the 
behaviour that is inactive, i.e.\ it performs no more basic actions and 
it does not terminate. 

\BTA\ has only one axiom.
This axiom is given in Table~\ref{axioms-BTA}.
\begin{table}[!t]
\caption{Axiom of \BTA} 
\label{axioms-BTA}
\begin{eqntbl}
\begin{axcol}
\pcc{x}{\Tau}{y} = \pcc{x}{\Tau}{x}                      & \axiom{T1}
\end{axcol}
\end{eqntbl}
\end{table}
Using the abbreviation introduced above, it can also be written as
follows: $\pcc{x}{\Tau}{y} = \Tau \bapf x$.

Each closed \BTA\ term represents a finite thread, i.e.\ a thread with 
a finite upper bound to the number of basic actions that it can perform.
Infinite threads, i.e.\ threads without a finite upper bound to the
number of basic actions that it can perform, can be defined by means of 
a set of recursion equations (see e.g.~\cite{BM09k}).

A regular thread is a finite or infinite thread that can be defined by 
means of a finite set of recursion equations.
The behaviours produced by \PGA\ instruction sequences under execution 
are exactly the behaviours modeled by regular threads.

Two infinite threads are considered identical if their approximations up 
to any finite depth are identical.
The approximation up to depth $n$ of a thread models the behaviour that 
differs from the behaviour modeled by the thread in that it will become
inactive after it has performed $n$ actions unless it would terminate at
this point.
AIP (Approximation Induction Principle) is a conditional equation that
formalizes the above-mentioned view on infinite threads.
In AIP, the approximation up to depth $n$ is phrased in terms of the
unary \emph{projection} operator $\funct{\proj{n}}{\Thr}{\Thr}$.

The axioms for the projection operators and AIP are given in
Table~\ref{axioms-BTAinf}.
\begin{table}[!t]
\caption{Axioms for the projection operators and AIP} 
\label{axioms-BTAinf}
\begin{eqntbl}
\begin{axcol}
\proj{0}(x) = \DeadEnd                                  & \axiom{PR1} \\
\proj{n+1}(\DeadEnd) = \DeadEnd                         & \axiom{PR2} \\
\proj{n+1}(\Stop) = \Stop                               & \axiom{PR3} \\
\proj{n+1}(\pcc{x}{\alpha}{y}) =
\pcc{\proj{n}(x)}{\alpha}{\proj{n}(y)}                  & \axiom{PR4}
\eqnsep
\LAND{n \geq 0} \proj{n}(x) = \proj{n}(y) \Limpl x = y  & \axiom{AIP}
\end{axcol}
\end{eqntbl}
\end{table}
In this table, $\alpha$ stands for an arbitrary basic action from 
$\BActTau$ and $n$ stands for an arbitrary natural number from $\Nat$.
We write \BTAinf\ for \BTA\ extended with the projection operators, 
the axioms for the projection operators, and AIP.

By AIP, we have to deal in \BTAinf\ with conditional equational formulas 
with a countably infinite number of premises.
Therefore, infinitary conditional equational logic is used in deriving 
equations from the axioms of \BTAinf.
A complete inference system for infinitary conditional equational logic 
can be found in, for example, \cite{GV93}.

The depth of a finite thread is the maximum number of basic actions that 
the thread can perform before it terminates or becomes inactive.
We define the function $\depth$ that assigns to each closed \BTA\ term
the depth of the finite thread that it represents:
\begin{ldispl}
\depth(\Stop) = 0\;, \\
\depth(\DeadEnd) = 0\;, \\
\depth(\pcc{t}{\alpha}{t'}) = \max \set{\depth(t),\depth(t')} + 1\;.
\end{ldispl}%

\section{Thread Extraction and Behavioural Congruence}
\label{sect-TE-BC}

In this section, we make precise in the setting of \BTAinf\ which 
behaviours are produced by \PGA\ instruction sequences under 
execution and introduce the notion of behavioural congruence on \PGA\ 
instruction sequences.

To make precise which behaviours are produced by \PGA\ instruction 
sequences under execution, we introduce an operator $\extr{\ph}$.
For each closed \PGA\ term $t$, $\extr{t}$ represents the thread that
models the behaviour produced by the instruction sequence represented 
by $t$ under execution.

Formally, we combine \PGA\ with \BTAinf\ and extend the combination with 
the \emph{thread extraction} operator $\funct{\extr{\ph}}{\InSeq}{\Thr}$ 
and the axioms given in Table~\ref{axioms-thread-extr}.%
\begin{table}[!t]
\caption{Axioms for the thread extraction operator} 
\label{axioms-thread-extr}
\begin{eqntbl}
\begin{axcol}
\extr{a} = a \bapf \DeadEnd                            & \axiom{TE1}  \\
\extr{a \conc X} = a \bapf \extr{X}                    & \axiom{TE2}  \\
\extr{\ptst{a}} = a \bapf \DeadEnd                     & \axiom{TE3}  \\
\extr{\ptst{a} \conc X} = \pcc{\extr{X}}{a}{\extr{\fjmp{2} \conc X}}
                                                       & \axiom{TE4}  \\
\extr{\ntst{a}} = a \bapf \DeadEnd                     & \axiom{TE5}  \\
\extr{\ntst{a} \conc X} = \pcc{\extr{\fjmp{2} \conc X}}{a}{\extr{X}}
                                                       & \axiom{TE6}
\end{axcol}
\qquad
\begin{axcol}
\extr{\fjmp{l}} = \DeadEnd                             & \axiom{TE7}  \\
\extr{\fjmp{0} \conc X} = \DeadEnd                     & \axiom{TE8}  \\
\extr{\fjmp{1} \conc X} = \extr{X}                     & \axiom{TE9}  \\
\extr{\fjmp{l+2} \conc u} = \DeadEnd                   & \axiom{TE10} \\
\extr{\fjmp{l+2} \conc u \conc X} = \extr{\fjmp{l+1} \conc X}
                                                       & \axiom{TE11} \\
\extr{\halt} = \Stop                                   & \axiom{TE12} \\
\extr{\halt \conc X} = \Stop                           & \axiom{TE13}
\end{axcol}
\end{eqntbl}
\end{table}
In this table, 
$a$ stands for an arbitrary basic instruction from $\BInstr$, 
$u$ stands for an arbitrary primitive instruction from $\PInstr$, and 
$l$ stands for an arbitrary natural number from $\Nat$.
We write \PGABTA\ for the combination of \PGA\ and \BTAinf\ extended 
with the thread extraction operator and the axioms for the thread 
extraction operator.

If a closed \PGA\ term $t$ represents an instruction sequence that
starts with an infinite chain of forward jumps, then TE9 and TE11 can 
be applied to $\extr{t}$ infinitely often without ever showing that a 
basic action is performed.
In this case, we have to do with inaction and, being consistent with 
that, $\extr{t} = \DeadEnd$ is derivable from the axioms of \PGA\ and 
TE1--TE13.
By contrast, $\extr{t} = \DeadEnd$ is not derivable from the axioms of 
\PGAisc\ and TE1--TE13.
However, if closed \PGA\ terms $t$ and $t'$ represent instruction 
sequences in which no infinite chains of forward jumps occur, then 
$t = t'$ is derivable from the axioms of \PGA\ only if 
$\extr{t} = \extr{t'}$ is derivable from the axioms of \PGAisc\ and 
TE1--TE13.

If a closed \PGA\ term $t$ represents an infinite instruction 
sequence, then we can extract the approximations of the thread modeling
the behaviour produced by that instruction sequence under execution up 
to every finite depth: for each $n \in \Nat$, there exists a closed 
\BTA\ term $t''$ such that $\proj{n}(\extr{t}) = t''$ is derivable 
from the axioms of \PGA, TE1--TE13, the axioms of \BTA, and PR1--PR4.
If closed \PGA\ terms $t$ and $t'$ represent infinite instruction 
sequences that produce the same behaviour under execution, then this can
be proved using the following instance of AIP:
$\LAND{n \geq 0} \proj{n}(\extr{t}) = \proj{n}(\extr{t'}) \Limpl
 \extr{t} = \extr{t'}$.

The following proposition, proved in~\cite{BM12b}, puts the 
expressiveness of \PGA\ in terms of producible behaviours.
\begin{proposition}
\label{proposition-expr}
Let $\mathcal{M}$ be a model of \PGABTA.
Then, for each element $p$ from the domain associated with the sort 
$\Thr$ in $\mathcal{M}$, there exists a closed \PGA\ term $t$ such that 
$p$ is the interpretation of $\extr{t}$ in $\mathcal{M}$ iff $p$ is a 
component of the solution of a finite set of recursion equations 
$\set{V = t_V \where V \in \mathcal{V}}$, where $\mathcal{V}$ is a set 
of variables of sort $\Thr$ and each $t_V$ is a \BTA\ term that is not 
a variable and contains only variables from~$\mathcal{V}$.
\end{proposition}
More results on the expressiveness of \PGA\ can be found 
in~\cite{BM12b}.

\PGA\ instruction sequences are behaviourally equivalent if they 
produce the same behaviour under execution.
Behavioural equivalence is not a congruence.
Instruction sequences are behaviourally congruent if they produce the
same behaviour irrespective of the way they are entered and the way
they are left.

Let $t$ and $t'$ be closed \PGA\ terms.
Then:
\begin{itemize}
\item
$t$ and $t'$ are \emph{behaviourally equivalent}, 
written $t \beqv t'$, if $\extr{t} = \extr{t'}$ is derivable
from the axioms of \PGABTA;
\item
$t$ and $t'$ are \emph{behaviourally congruent}, written 
$t \bcong t'$, if, for each $l,n \in \Nat$,
$\fjmp{l} \conc t \conc \halt^n \beqv \fjmp{l} \conc t' \conc \halt^n$.%
\footnote
{We use the convention that $t \conc {t'}^0$ stands for $t$.}
\end{itemize}
Behavioural congruence is the largest congruence contained in
behavioural equivalence.

\section{The Case of Instructions for Turing Tapes}
\label{sect-PGAtt}

In this section, we present the instantiation of \PGA\ in which the 
instructions taken as basic instructions are all possible instructions 
to read out or alter the content of the cell of a Turing tape under the 
tape's head and to optionally move the head in either direction by one 
cell.

The instructions concerned are meant for Turing tapes of which each cell 
contains a symbol from the input alphabet $\IAlph$ or the symbol 
$\Blank$.
Turing proposed computing machines with a tape of which each cell 
contains a symbol from a finite alphabet, the so-called tape alphabet, 
that includes the input alphabet $\IAlph$ and the symbol $\Blank$ 
(see~\cite{Tur37a}).%
\footnote
{In many publications in which Turing machines are defined, the input 
 alphabet may even be any non-empty finite set of symbols.}
The tape alphabet may differ from one machine to another.
The choice between the tape alphabet $\TAlph$ and any tape alphabet that 
includes $\TAlph$ is rather arbitrary because it has no effect on 
the computability and the order-of-magnitude time complexity of partial 
functions from ${(\IAlph^*)}^n$ to $\IAlph^*$ ($n \geq 0$).
We have chosen for the tape alphabet $\TAlph$ because it allows of 
presenting part of the material to come in a more comprehensible manner.

In the present instantiation of \PGA, it is assumed that a fixed but 
arbitrary set $\Foci$ of \emph{foci} has been given.
Foci serve as names of Turing tapes.

The set of basic instructions used in this instantiation consists of the 
following:
\begin{itemize}
\item[]
for each $\funct{p}{\TAlph}{\IAlph}$, 
$\funct{q}{\TAlph}{\TAlph}$, $d \in \Dir$, and $f \in \Foci$, 
a \emph{basic Turing-tape instruction} $f.\mtt{p}{(q,d)}$.
\end{itemize}
We write $\BInstrtt$ for this set.

Each basic Turing-tape instruction consists of two parts separated 
by a dot.
The part on the left-hand side of the dot plays the role of the name of 
a Turing tape and the part on the right-hand side of the dot plays 
the role of an operation to be carried out on the named Turing tape 
when the instruction is executed.
The intuition is basically that carrying out the operation concerned 
produces as a reply $0$ or $1$ depending on the content of the cell 
under the head of the named Turing tape, modifies the content of this 
cell, and optionally moves the head in either direction by one cell.
More precisely, the execution of a basic Turing-tape instruction 
$f.\mtt{p}{(q,d)}$ has the following effects:
\begin{itemize}
\item
if the content of the cell under the head of the Turing tape named $f$ 
is $b$ when the execution of $f.\mtt{p}{(q,d)}$ starts, then the reply 
produced on termination of the execution of $f.\mtt{p}{(q,d)}$ is 
$p(b)$;
\item
if the content of the cell under the head of the Turing tape named $f$ 
is $b$ when the execution of $f.\mtt{p}{(q,d)}$ starts, then the content 
of this cell is $q(b)$ when the execution of $f.\mtt{p}{(q,d)}$ 
terminates;
\item
if the cell under the head of the Turing tape named $f$ is the $i$th 
cell of the tape when the execution of $f.\mtt{p}{(q,d)}$ starts, then 
the cell under the head of this Turing tape is the $\max(i+d,1)$th cell 
when the execution of $f.\mtt{p}{(q,d)}$ terminates.
\end{itemize}
The execution of $f.\mtt{p}{(q,d)}$ has no effect on Turing tapes other 
than the one named $f$.

We write $\PGABTAtt$ for \PGABTA\ with $\BInstr$ instantiated by 
$\BInstrtt$.
Notice that $\PGABTAtt$ is itself parameterized by a set of foci.
 
Some functions from $\TAlph$ to $\TAlph$ are:
\begin{itemize}
\item
the function $\FTest$, \,satisfying 
$\FTest(\False) = \True$ and $\FTest(\True) = \False$ \,and 
$\FTest(\Blank) = \False$;
\item
the function $\TTest$, \,satisfying 
$\TTest(\False) = \False$ and $\TTest(\True) = \True$ \,and 
$\TTest(\Blank) = \False$;
\item
the function $\BTest$, satisfying 
$\BTest(\False) = \False$ and $\BTest(\True) = \False$ and 
$\BTest(\Blank) = \True$;
\item
the function $\FFunc$, \,satisfying 
$\FFunc(\False) = \False$ and $\FFunc(\True) = \False$ \,and 
$\FFunc(\Blank) = \False$;
\item
the function $\TFunc$, \,satisfying 
$\TFunc(\False) = \True$ and $\TFunc(\True) = \True$ \,and 
$\FFunc(\Blank) = \True$;
\item
the function $\BFunc$, satisfying 
$\TFunc(\False) = \Blank$ and $\TFunc(\True) = \Blank$ and 
$\FFunc(\Blank) = \Blank$;
\item
the function $\IFunc$, \,satisfying 
$\IFunc(\False) = \False$ and $\IFunc(\True) = \True$ \,and 
$\FFunc(\Blank) = \Blank$;
\item
the function $\CFunc$, \,satisfying 
$\CFunc(\False) = \True$ and $\CFunc(\True) = \False$ \,and 
$\FFunc(\Blank) = \Blank$.
\end{itemize}
The first five of these functions are also functions from $\TAlph$ to 
$\IAlph$.

For some instances of $\mtt{p}{(q,d)}$, we introduce a special notation.
We write:
\begin{ldispl}
\begin{array}[t]{@{}l@{\;\;\mathrm{for}\;\;}l@{}}
\tsttt{\False} & \mtt{\FTest}{(\IFunc,0)}\;, \\
\tsttt{\True}  & \mtt{\TTest}{(\IFunc,0)}\;, \\
\tsttt{\Blank} & \mtt{\BTest}{(\IFunc,0)}\;,
\end{array}
\qquad
\begin{array}[t]{@{}l@{\;\;\mathrm{for}\;\;}l@{}}
\settt{\False}{d} & \mtt{\TFunc}{(\FFunc,d)}\;, \\
\settt{\True}{d}  & \mtt{\TFunc}{(\TFunc,d)}\;, \\
\settt{\Blank}{d} & \mtt{\TFunc}{(\BFunc,d)}\;,
\end{array}
\qquad
\begin{array}[t]{@{}l@{\;\;\mathrm{for}\;\;}l@{}}
\skptt{d} & \mtt{\TFunc}{(\IFunc,d)}\;,
\end{array}
\end{ldispl}
where $d \in \Dir$.

\section{Turing-Tape Families}
\label{sect-TTFA}

\PGA\ instruction sequences under execution may interact with the named 
services from a family of services provided by their execution 
environment.
In applications of \PGA, the services provided by an execution 
environment vary from Boolean registers via Turing tapes to random 
access memories of actual computers.%
\footnote
{A Boolean register consists of a single cell that contains a symbol
 from the alphabet $\IAlph$.
 Carrying out an operation on a Boolean register produces as a reply
 $0$ or $1$, depending on the content of the cell, and/or modifies the 
 content of the cell.}   
In this section, we consider service families in which the services are 
Turing tapes and introduce an algebraic theory of Turing-tape families 
called \TTFA\ (Turing-Tape Family Algebra).

A \emph{Turing-tape state} is a pair $(\tau,i)$, 
where $\funct{\tau}{\Natpos}{\TAlph}$ and $i \in \Natpos$, 
satisfying the condition that, for some $j \in \Natpos$, 
for all $k \in \Nat$, $\tau(j + k) = \Blank$. 
We write $\TStates$ for the set of all Turing-tape states.

Let $(\tau,i)$ be a Turing-tape state.
Then, for all $j \in \Natpos$, $\tau(j)$ is the content of the $j$th 
cell of the Turing tape concerned and the $i$th cell is the cell under 
its head.

Our Turing tapes are one-way infinite tapes.
Turing proposed computing machine with two-way infinite tapes 
(see~\cite{Tur37a}).
In many publications in which Turing machine are defined, Turing 
machines are a variant of Turing's computing machines with one or more 
one-way infinite tapes (cf.\ the 
textbooks~\mbox{\cite{AHU74a,AB09a,Gol08a,HS11a,Koz06a,Sip13a}}).
The choice between one-way infinite tapes and two-way infinite tapes is 
rather arbitrary because it has no effect on the computability and the 
order-of-magnitude time complexity  of partial functions from 
${(\IAlph^*)}^n$ to $\IAlph^*$ ($n \geq 0$).
We have chosen for one-way infinite tapes because it allows of 
presenting part of the material to come in a more comprehensible manner.

In Section~\ref{sect-TSI}, we will use the notation $\fncv{\tau}{i}{b}$.
For each $\funct{\tau}{\Natpos}{\TAlph}$, $i \in \Natpos$, and 
$b \in \TAlph$, $\fncv{\tau}{i}{b}$ is defined as follows: 
$\fncv{\tau}{i}{b}(i) = b$ and, for all $j \in \Natpos$ with $j \neq i$,
$\fncv{\tau}{i}{b}(j) = \tau(j)$.

In \TTFA, as in $\PGABTAtt$, it is assumed that a fixed but arbitrary 
set $\Foci$ of foci has been given.

\TTFA\ has one sort: the sort $\TTapeFam$ of 
\emph{Turing-tape families}.
To build terms of sort $\TTapeFam$, \TTFA\ has the following constants 
and operators:
\begin{itemize}
\item
the
\emph{empty Turing-tape family} constant 
$\const{\emptysf}{\TTapeFam}$;
\item
for each $f \in \Foci$ and $s \in \TStates \union \set{\Div}$, 
the \emph{singleton Turing-tape family} constant
$\const{f.\stt{s}}{\TTapeFam}$;
\item
the binary \emph{Turing-tape family composition} operator
$\funct{\ph \sfcomp \ph}{\TTapeFam \x \TTapeFam}{\TTapeFam}$;
\item
for each $F \subseteq \Foci$, 
the unary \emph{encapsulation} operator 
$\funct{\encap{F}}{\TTapeFam}{\TTapeFam}$.
\end{itemize}
We assume that there are infinitely many variables of sort $\TTapeFam$,
including $u,v,w$.
We use infix notation for the Turing-tape family composition 
operator.

The Turing-tape family denoted by $\emptysf$ is the empty Turing-tape 
family.
The Turing-tape family denoted by a closed term of the form $f.\stt{s}$, 
where $s \in \TStates$, consists of one named Turing tape only, the 
Turing tape concerned is an operative Turing tape named $f$ whose state 
is $s$.
The Turing-tape family denoted by a closed term of the form 
$f.\stt{\Div}$ consists of one named Turing tape only, the Turing tape 
concerned is an inoperative Turing tape named $f$.
The Turing-tape family denoted by a closed term of the form
$t \sfcomp t'$ consists of all named Turing tapes that belong to 
either the Turing-tape family denoted by $t$ or the Turing-tape family 
denoted by $t'$.
In the case where a named Turing tape from the Turing-tape 
family denoted by $t$ and a named Turing tape from the Turing-tape 
family denoted by $t'$ have the same name, they collapse to 
an inoperative Turing tape with the name concerned.
The Turing-tape family denoted by a closed term of the form 
$\encap{F}(t)$ consists of all named Turing tapes with a name not 
in $F$ that belong to the Turing-tape family denoted by $t$.

An inoperative Turing tape can be viewed as a Turing tape whose state 
is unavailable.
Carrying out an operation on an inoperative Turing tape is impossible.

The axioms of \TTFA\ are given in Table~\ref{axioms-BRFA}.%
\begin{table}[!t]
\caption{Axioms of \TTFA} 
\label{axioms-BRFA}
\begin{eqntbl}
\begin{axcol}
u \sfcomp \emptysf = u                                 & \axiom{TTFC1} \\
u \sfcomp v = v \sfcomp u                              & \axiom{TTFC2} \\
(u \sfcomp v) \sfcomp w = u \sfcomp (v \sfcomp w)      & \axiom{TTFC3} \\
f.\stt{s} \sfcomp f.\stt{s'} = f.\stt{\Div}            & \axiom{TTFC4}
\end{axcol}
\qquad
\begin{saxcol}
\encap{F}(\emptysf) = \emptysf                       & & \axiom{TTFE1} \\
\encap{F}(f.\stt{s}) = \emptysf      & \mif f \in F    & \axiom{TTFE2} \\
\encap{F}(f.\stt{s}) = f.\stt{s}     & \mif f \notin F & \axiom{TTFE3} \\
\multicolumn{2}{@{}l@{\quad}}
 {\encap{F}(u \sfcomp v) =
  \encap{F}(u) \sfcomp \encap{F}(v)}                   & \axiom{TTFE4}
\end{saxcol}
\end{eqntbl}
\end{table}
In this table, $f$ stands for an arbitrary focus from $\Foci$, 
$F$ stands for an arbitrary subset of $\Foci$, and 
$s$ and $s'$ stand for arbitrary members of $\TStates \union \set{\Div}$.
These axioms simply formalize the informal explanation given
above.

The following proposition, proved in~\cite{BM12b}, is a representation 
result for closed \TTFA\ terms.
\begin{proposition}
\label{proposition-represent}
For all closed \TTFA\ terms $t$, for all $f \in \Foci$, either 
$t = \encap{\set{f}}(t)$ is derivable from the axioms of \TTFA\ or there 
exists an $s \in \TStates \union \set{\Div}$ such that 
$t = f.\stt{s} \sfcomp \encap{\set{f}}(t)$ is derivable from the axioms 
of \TTFA.
\end{proposition}

In Section~\ref{sect-comput-boolstring-fnc}, we will use the notation 
$\Sfcomp{i = 1}{n} t_i$.
For each $i \in \Natpos$, let $t_i$ be a terms of sort $\TTapeFam$.
Then, for each $n \in \Natpos$, the term $\Sfcomp{i = 1}{n} t_i$ is 
defined by induction on $n$ as follows: $\Sfcomp{i = 1}{1} t_i = t_1$ 
and $\Sfcomp{i = 1}{n+1} t_i = \Sfcomp{i = 1}{n} t_i \sfcomp t_{n+1}$.
We use the convention that $\Sfcomp{i = 1}{n} t_i$ stands for $\emptysf$
if $n = 0$.

\section{Interaction of Threads with Turing Tapes}
\label{sect-TSI}

If instructions from $\BInstrtt$ are taken as basic instructions, a 
\PGA\ instruction sequence under execution may interact with named 
Turing tapes from a family of Turing tapes provided by its execution 
environment.
In line with this kind of interaction, a thread may perform a basic 
action basically for the purpose of modifying the content of a named 
Turing tape or receiving a reply value that depends on the content 
of a named Turing tape.
In this section, we introduce related operators.

We combine $\PGABTA(\BInstrtt)$ with \TTFA\ and extend the combination 
with the following operators for interaction of threads with Turing
tapes:
\begin{itemize}
\item
the binary \emph{use} operator
$\funct{\ph \sfuse \ph}{\Thr \x \TTapeFam}{\Thr}$;
\item
the binary \emph{apply} operator
$\funct{\ph \sfapply \ph}{\Thr \x \TTapeFam}{\TTapeFam}$;
\item
the unary \emph{abstraction} operator 
$\funct{\abstr{\Tau}}{\Thr}{\Thr}$;
\end{itemize}
and the axioms given in Tables~\ref{axioms-use-apply}.%
\footnote
{We write $t[t'/x]$ for the result of substituting term $t'$ for 
variable $x$ in term $t$.}
\begin{table}[!t]
\caption{Axioms for the use, apply and abstraction operator} 
\label{axioms-use-apply}
\begin{eqntbl}
\begin{saxcol}
\Stop  \sfuse u = \Stop                                & & \axiom{U1} \\
\DeadEnd \sfuse u = \DeadEnd                           & & \axiom{U2} \\
(\Tau \bapf x) \sfuse u = \Tau \bapf (x \sfuse u)      & & \axiom{U3} \\
\multicolumn{2}{@{}l@{}}{
(\pcc{x}{f.\mtt{p}{(q,d)}}{y}) \sfuse \encap{\set{f}}(u) =
\pcc{(x \sfuse \encap{\set{f}}(u))}
 {f.\mtt{p}{(q,d)}}{(y \sfuse \encap{\set{f}}(u))}     } & \axiom{U4} \\
(\pcc{x}{f.\mtt{p}{(q,d)}}{y}) \sfuse 
(f.\stt{\tau,i} \sfcomp \encap{\set{f}}(u)) \\ \quad {} = 
\Tau \bapf (x \sfuse (f.\stt{\fncv{\tau}{i}{q(\tau(i))},\max(i+d,1)} 
 \sfcomp \encap{\set{f}}(u))) & \mif p(\tau(i)) = \True  & \axiom{U5} \\
(\pcc{x}{f.\mtt{p}{(q,d)}}{y}) \sfuse 
(f.\stt{\tau,i} \sfcomp \encap{\set{f}}(u)) \\ \quad {} =
\Tau \bapf (y \sfuse (f.\stt{\fncv{\tau}{i}{q(\tau(i))},\max(i+d,1)} 
 \sfcomp \encap{\set{f}}(u))) & \mif p(\tau(i)) = \False & \axiom{U6} \\
(\pcc{x}{f.\mtt{p}{(q,d)}}{y}) \sfuse 
(f.\stt{\Div} \sfcomp \encap{\set{f}}(u)) = \DeadEnd                                          
                                                       & & \axiom{U7} \\
\proj{n}(x \sfuse u) = \proj{n}(x) \sfuse u            & & \axiom{U8} 
\eqnsep
\Stop  \sfapply u = u                                  & & \axiom{A1} \\
\DeadEnd \sfapply u = \emptysf                         & & \axiom{A2} \\
(\Tau \bapf x) \sfapply u = \Tau \bapf (x \sfapply u)  & & \axiom{A3} \\
(\pcc{x}{f.\mtt{p}{(q,d)}}{y}) \sfapply \encap{\set{f}}(u) = \emptysf
                                                       & & \axiom{A4} \\
(\pcc{x}{f.\mtt{p}{(q,d)}}{y}) \sfapply 
(f.\stt{\tau,i} \sfcomp \encap{\set{f}}(u)) \\ \qquad {} =
x \sfapply (f.\stt{\fncv{\tau}{i}{q(\tau(i))},\max(i+d,1)} 
  \sfcomp \encap{\set{f}}(u)) & \mif p(\tau(i)) = \True  & \axiom{A5} \\
(\pcc{x}{f.\mtt{p}{(q,d)}}{y}) \sfapply 
(f.\stt{\tau,i} \sfcomp \encap{\set{f}}(u)) \\ \qquad {} =
y \sfapply (f.\stt{\fncv{\tau}{i}{q(\tau(i))},\max(i+d,1)} 
  \sfcomp \encap{\set{f}}(u)) & \mif p(\tau(i)) = \False & \axiom{A6} \\
(\pcc{x}{f.\mtt{p}{(q,d)}}{y}) \sfapply 
(f.\stt{\Div} \sfcomp \encap{\set{f}}(u)) = \emptysf   & & \axiom{A7} \\
\LAND{k \geq n} t[\proj{k}(x)/z] = t'[\proj{k}(y)/z] \Limpl 
t[x/z] = t'[y/z]                                        & & \axiom{A8}
\eqnsep
\abstr{\Tau}(\Stop) = \Stop                            & & \axiom{C1} \\
\abstr{\Tau}(\DeadEnd) = \DeadEnd                      & & \axiom{C2} \\
\abstr{\Tau}(\Tau \bapf x) = \abstr{\Tau}(x)           & & \axiom{C3} \\
\abstr{\Tau}(\pcc{x}{f.\mtt{p}{(q,d)}}{y}) =
\pcc{\abstr{\Tau}(x)}{f.\mtt{p}{(q,d)}}{\abstr{\Tau}(y)}   & & \axiom{C4} \\
\LAND{k \geq 0}{}
 \abstr{\Tau}(\proj{k}(x)) = \abstr{\Tau}(\proj{k}(y)) \Limpl
\abstr{\Tau}(x) = \abstr{\Tau}(y)                      & & \axiom{C5}
\end{saxcol}
\end{eqntbl}
\end{table}
In these tables, $f$ stands for an arbitrary focus from $\Foci$, 
$p$ stands for an arbitrary function from $\TAlph$ to $\IAlph$,
$q$ stands for an arbitrary function from $\TAlph$ to $\TAlph$, 
$d$ stands for an arbitrary member of $\Dir$,
$\tau$ stands for an arbitrary function from $\Natpos$ to $\TAlph$, 
$i$~stands for an arbitrary natural number from $\Natpos$,
$n$ stands for an arbitrary natural number from $\Nat$, and
$t$ and $t'$ stand for arbitrary terms of sort $\TTapeFam$.
We use infix notation for the use and apply operators.
We write \PGABTATTI\ for the combination of \PGABTAtt\ and \TTFA\ 
extended with the use operator, the apply operator, the abstraction 
operator, and the axioms for these operators.

Axioms U1--U7 and A1--A7 formalize the informal explanation of the use 
operator and the apply operator given below and in addition stipulate 
what is the result of apply if an unavailable focus is involved~(A4) and 
what is the result of use and apply if an inoperative Turing tape 
is involved (U7 and A7).
Axioms U8 and A8 allow of reasoning about infinite threads, and 
therefore about the behaviour produced by infinite instruction sequences 
under execution, in the context of use and apply, respectively.

On interaction between a thread and a Turing tape, the thread 
affects the Turing tape and the Turing tape affects the 
thread.
The use operator concerns the effects of Turing tapes on threads 
and the apply operator concerns the effects of threads on Turing tapes.
The thread denoted by a closed term of the form $t \sfuse t'$ and the
Turing-tape family denoted by a closed term of the form
$t \sfapply t'$ are the thread and Turing-tape family, 
respectively, that result from carrying out the operation that is part 
of each basic action performed by the thread denoted by $t$ on the 
Turing tape in the Turing-tape family denoted by $t'$ with the 
focus that is part of the basic action as its name.
When the operation that is part of a basic action performed by a thread 
is carried out on a Turing tape, the content of the Turing tape is 
modified according to the operation concerned and the thread is affected 
as follows: the basic action turns into the internal action $\Tau$ and 
the two ways to proceed reduce to one on the basis of the reply value 
produced according to the operation concerned.

With the use operator the internal action $\Tau$ is left as a trace of 
each basic action that has led to carrying out an operation on a Turing
tape.
The abstraction operator serves to abstract fully from such internal 
activity by concealing $\Tau$.
Axioms C1--C4 formalizes the concealment of $\Tau$.
Axiom C5 allows of reasoning about infinite threads in the context of 
abstraction.

The following two theorems are elimination results for closed \PGABTATTI\ 
terms. \sloppy
\begin{theorem}
\label{theorem-elim-use}
For all closed \PGABTATTI\ terms $t$ of sort $\Thr$ in which all 
subterms of sort $\InSeq$ are repetition-free, there exists a closed 
\PGABTAtt\ term $t'$ of sort $\Thr$ such that $t = t'$ is derivable from 
the axioms of \PGABTATTI.
\end{theorem}
\begin{proof}
It is easy to prove by structural induction that, for all closed 
rep\-etition-free \PGABTAtt\ terms $s$ of sort $\InSeq$, there exists a 
closed \PGABTAtt\ term $s'$ of sort $\Thr$ such that $\extr{s} = s'$ is 
derivable from the axioms of \PGABTAtt.
Therefore, it is sufficient to prove the proposition for all closed 
\PGABTATTI\ terms $t$ of sort $\Thr$ in which no subterms of sort 
$\InSeq$ occur.
This is proved similarly to part~(1) of Theorem~3.1 from~\cite{BM12b}. 
\qed
\end{proof}
\begin{theorem}
\label{theorem-elim-apply}
For all closed \PGABTATTI\ terms $t$ of sort $\TTapeFam$ in which all 
subterms of sort $\InSeq$ are repetition-free, there exists a closed 
\PGABTAtt\ term $t'$ of sort $\TTapeFam$ such that $t = t'$ is derivable 
from the axioms of \PGABTATTI.
\end{theorem}
\begin{proof}
As in the proof of Theorem~\ref{theorem-elim-use}, it is sufficient to 
prove the proposition for all closed \PGABTATTI\ terms $t$ of sort 
$\TTapeFam$ in which no subterms of sort $\InSeq$ occur.
This is proved similarly to part~(2) of Theorem~3.1 from~\cite{BM12b}. 
\qed
\end{proof}

\section{Computing Partial Functions from ${(\IAlph^*)}^n$ to $\IAlph^*$}
\label{sect-comput-boolstring-fnc}

In this section, we make precise in the setting of the algebraic theory
\PGABTATTI\ what it means that a given instruction sequence computes a 
given partial function from ${(\IAlph^*)}^n$ to $\IAlph^*$ 
($n \in \Nat$).

We write $\Focitt{k}$, where $k \in \Natpos$, for the set 
$\set{\ftt{i} \where 1 \leq i \leq k}$ of foci.
We write \PGABTATTIt{k}\ for \PGABTATTI\ with $\Foci$ instantiated by 
$\Focitt{k}$.

Below, we use the function
$\funct{\ctt}
  {\set{\funct{\tau}{\Natpos}{\TAlph} \where (\tau,1) \in \TStates}}
  {\TAlph^*}$ 
for extracting the content of a Turing tape.
This function is defined as follows: 
\begin{itemize}
\item[]
for all $n \in \Natpos$, for all $b_1,\ldots,b_n \in \TAlph$, \\
$\ctt(\tau) = b_1\,\ldots\,b_n$ iff 
$\tau(i) = b_i$ for all $i \leq n$,
$\tau(i) = \Blank$ for all $i > n$, and
$\tau(n) \neq \Blank$;
\item[]
$\ctt(\tau) = \epsilon\phantom{_1\,\ldots\,b_n}\hsp{-.075}$ iff 
$\tau(i) = \Blank$ for all $i \geq 1$.%
\footnote
{We write $\epsilon$ for the empty string.} 
\end{itemize}

Let $k \in \Natpos$,  
let $t$ be a closed \PGABTATTIt{k}\ term of sort $\InSeq$,
let $n \in \Nat$, let $\pfunct{F}{{(\IAlph^*)}^n}{\IAlph^*}$,%
\footnote
{We write $\pfunct{f}{A}{B}$ to indicate that $f$ is a partial function
 from $A$ to $B$.}
and
let $\funct{T}{\Nat}{\Nat}$.
Then $t$ \emph{computes $F$ with $k$ tapes in time $T$} if:
\begin{itemize}
\item
for all $w_1,\ldots,w_n \in \IAlph^*$ such that $F(w_1,\ldots,w_n)$ is
defined, \\
there exist $(\tau'_1,i_1),\ldots,(\tau'_{k-1},i_{k-1}) \in \TStates$ 
such that
\begin{ldispl}
\extr{t} \sfapply \Sfcomp{j = 1}{k} \ftt{j}.\stt{\tau_j,1} = 
\Sfcomp{j = 1}{k-1} \ftt{j}.\stt{\tau'_j,i_j} \sfcomp 
\ftt{k}.\stt{\tau'_k,1}\;,
\eqnsep
\depth(\extr{t} \sfuse \Sfcomp{j = 1}{k} \ftt{j}.\stt{\tau_j,1}) \leq
T(\len(w_1) + \ldots + \len(w_n))\;, 
\end{ldispl}%
where \\
$\tau_1$ is the unique $\funct{\tau}{\Natpos}{\TAlph}$ with 
$(\tau,1) \in \TStates$ and 
$\ctt(\tau) = w_1\,\Blank\,\ldots\,\Blank\,w_n$, \\
for $j \neq 1$, 
$\tau_j$ is the unique $\funct{\tau}{\Natpos}{\TAlph}$ with
$(\tau,1) \in \TStates$ and $\ctt(\tau) = \epsilon$, \\
$\tau'_k$ is the unique $\funct{\tau}{\Natpos}{\TAlph}$ with
$(\tau,1) \in \TStates$ and $\ctt(\tau) = F(w_1,\ldots,w_n)$; 
\item
for all $w_1,\ldots,w_n \in \IAlph^*$ such that $F(w_1,\ldots,w_n)$ is
undefined,
\begin{ldispl}
\extr{t} \sfapply \Sfcomp{j = 1}{k} \ftt{j}.\stt{\tau_j,1} = \emptysf\;,
\end{ldispl}%
where \\
$\tau_1$ is the unique $\funct{\tau}{\Natpos}{\TAlph}$ with 
$(\tau,1) \in \TStates$ and 
$\ctt(\tau) = w_1\,\Blank\,\ldots\,\Blank\,w_n$, \\
for $j \neq 1$, 
$\tau_j$ is the unique $\funct{\tau}{\Natpos}{\TAlph}$ with
$(\tau,1) \in \TStates$ and $\ctt(\tau) = \epsilon$.
\end{itemize}
We say that $t$ \emph{computes $F$ in time $T$} if there exists a 
$k \in \Natpos$ such that $t$ computes $F$ with $k$ tapes in time $T$, 
and we say that $t$ \emph{computes $F$} if there exists a 
$\funct{T}{\Nat}{\Nat}$ such that $t$ computes $F$ in time $T$.

With the above definition, we can establish whether an instruction 
sequence of the kind considered in \PGABTATTIt{k}\ ($k \in \Natpos$) 
computes a given partial function from ${(\IAlph^*)}^n$ to $\IAlph^*$ 
($n \in \Nat$) by equational reasoning using the axioms of 
\PGABTATTIt{k}.

A \emph{single-tape Turing-machine program} is a closed \PGABTATTIt{1}\ 
term of sort $\InSeq$ that is of the form 
$(t_1 \conc \ldots \conc t_n)\rep$, where each $t_i$ is of the form
\begin{ldispl}
\tsttt{0} \conc \fjmp{3} \conc \settt{b_0}{d_0} \conc u_0 \conc {} \\
\tsttt{1} \conc \fjmp{3} \conc \settt{b_1}{d_1} \conc u_1 \conc {} \\
\tsttt{\Blank} \conc \fjmp{3} \conc \settt{b_\Blank}{d_\Blank} \conc
u_\Blank\;, 
\end{ldispl}
where $b_0,b_1,b_\Blank \in \TAlph$, $d_0,d_1,d_\Blank \in \Dir$, and
\begin{itemize}
\item[]
$u_0$ \,is of the form 
$\fjmp{l}$ with $l \in \set{12 \mul i + 9 \where 0 \leq i < n}$ or 
$\fjmp{0}$ or $\halt$,
\item[]
$u_1$ \,is of the form 
$\fjmp{l}$ with $l \in \set{12 \mul i + 5 \where 0 \leq i < n}$ or 
$\fjmp{0}$ or $\halt$,
\item[]
$u_\Blank$ is of the form 
$\fjmp{l}$ with $l \in \set{12 \mul i + 1 \where 0 \leq i < n}$ or 
$\fjmp{0}$ or $\halt$.
\end{itemize}

We refrain from defining a $k$-tape Turing-machine program 
(for $k > 1$), which is much more involved than defining a single-tape 
Turing-machine program.
However, we remark that the theorems given below go through for $k$-tape 
Turing-machine programs.

The following theorem is a result concerning the computational power of 
single-tape Turing-machine programs.
\begin{theorem}
\label{theorem-turing-completeness}
For each $\pfunct{F}{{(\IAlph^*)}^n}{\IAlph^*}$, there exists a 
single-tape Turing-machine program $t$ such that $t$ computes $F$ iff 
$F$ is Turing-computable.
\end{theorem}
\begin{proof}
For each $\pfunct{F}{{(\IAlph^*)}^n}{\IAlph^*}$, $F$ is 
Turing-computable iff there exists a Turing machine with a single 
semi-infinite tape and stay option that computes $F$.
There is an obvious one-to-one correspondence between the transition 
functions of such Turing machines and single-tape Turing-machine 
programs by which the Turing machines concerned can be simulated when 
they are applied to a single tape.
Hence, for each $\pfunct{F}{{(\IAlph^*)}^n}{\IAlph^*}$, there exists a 
single-tape Turing-machine program $t$ such that $t$ computes $F$ iff 
$F$ is Turing-computable.
\qed
\end{proof}

Below, we write $\TMPst$ for the set of all single-tape Turing-machine 
programs, and
$\poly$ for 
$\set{T \where 
 \funct{T}{\Nat}{\Nat} \Land T \mathrm{\,is\,a\,polynomial\,function}}$.

The following theorem is a result relating the complexity class 
$\mathbf{P}$ to the functions that can be computed by a single-tape 
Turing-machine program in polynomial time.
\begin{theorem}
\label{theorem-polynomial-time}
$\mathbf{P}$ is equal to the class of all functions 
$\funct{F}{\IAlph^*}{\IAlph}$ for which there exist an $t \in \TMPst$ 
and a $T \in \poly$ such that $t$ computes $F$ in time $T$.
\end{theorem}
\begin{proof}
This follows from the proof of Theorem~\ref{theorem-turing-completeness} 
and the fact that, if a function $\funct{F}{\IAlph^*}{\IAlph}$ is 
computed on a Turing machine in time $T$, then the one-to-one 
correspondence referred to in the proof of 
Theorem~\ref{theorem-turing-completeness} yields for this Turing machine 
a single-tape Turing-machine program that computes $F$ in a time 
of~$O(T)$.
\qed
\end{proof}

We think that Theorems~\ref{theorem-turing-completeness} 
and~\ref{theorem-polynomial-time} above provide evidence of the claim 
that \PGABTATTI\ is a suitable setting for the development of theory in 
areas such as computability and computational complexity.
Moreover, in this setting variations on Turing machines that have not 
attracted attention yet come into the picture and can be studied.

\section{A Turing-Machine Program Example}
\label{sect-example}

In this section, we give a simple example of a Turing-machine program. 
We consider the \emph{non-zeroness test} function 
$\funct{\NZT}{{(\IAlph^*)}^1}{\IAlph^*}$ defined by
\begin{ldispl}
\NZT(b_1\,\ldots\,b_n) = 0 \;\;\mathrm{if}\;\;
b_1 = 0 \;\mathrm{and}\; \ldots \;\mathrm{and}\; b_n = 0\;,
\\
\NZT(b_1\,\ldots\,b_n) = 1 \;\;\mathrm{if}\;\;
b_1 = 1 \;\;\,\mathrm{or}\;\; \ldots \;\;\,\mathrm{or}\;\; b_n = 1\;.
\end{ldispl}%
$\NZT$ models the function $\funct{\NZTN}{\Nat}{\Nat}$ defined by 
$\NZTN(0) = 0$ and $\NZTN(k + 1) = 1$ with respect to the binary 
representations of the natural numbers.

We define a Turing-machine program $\NZTIS$ that computes $\NZT$ as 
follows:
\begin{ldispl}
\NZTIS \deq 
(\ntst{\tsttt{0}} \conc \fjmp{3} \conc \settt{0}{1} \conc
 \fjmp{33} \conc {}
\\ \phantom{\NZTIS \deq (}
 \ntst{\tsttt{1}} \conc \fjmp{3} \conc \settt{1}{1} \conc
 \fjmp{29} \conc {}
\\ \phantom{\NZTIS \deq (}
 \ntst{\tsttt{\Blank}} \conc \fjmp{3} \conc \settt{\Blank}{{-}1} \conc
 \fjmp{1} \conc {}
\\ \phantom{\NZTIS \deq (}
 \ntst{\tsttt{0}} \conc \fjmp{3} \conc \settt{\Blank}{{-}1} \conc
 \fjmp{33} \conc {}
\\ \phantom{\NZTIS \deq (}
 \ntst{\tsttt{1}} \conc \fjmp{3} \conc \settt{\Blank}{{-}1} \conc
 \fjmp{5} \conc {}
\\ \phantom{\NZTIS \deq (}
 \ntst{\tsttt{\Blank}} \conc \fjmp{3} \conc \settt{0}{0} \conc
 \halt \conc {}
\\ \phantom{\NZTIS \deq (}
 \ntst{\tsttt{0}} \conc \fjmp{3} \conc \settt{\Blank}{{-}1} \conc
 \fjmp{33} \conc {}
\\ \phantom{\NZTIS \deq (}
 \ntst{\tsttt{1}} \conc \fjmp{3} \conc \settt{\Blank}{{-}1} \conc
 \fjmp{29} \conc {}
\\ \phantom{\NZTIS \deq (}
 \ntst{\tsttt{\Blank}} \conc \fjmp{3} \conc \settt{1}{0} \conc
 \halt)\rep\;.
\end{ldispl}%
First, the head is moved to the right cell-by-cell until the first cell 
whose content is $\Blank$ has been reached and after that the head is 
moved to the left by one cell.
Then, the head is moved to the left cell-by-cell until the first cell of 
the tape has been reached and on top of that the content of each cell 
that comes under the head is replaced by $\Blank$.
Finally, the content of the first cell is replaced by $1$ if at least 
one cell with content $1$ has been encountered during the moves to the 
left and the content of the first cell is replaced by $0$ if no cell 
with content $1$ has been encountered during the moves to the left.

Because Turing-machine programs closely resemble the transition 
functions of Turing machines, they have built-in inefficiencies.
We use $\NZTIS$ to illustrate this.
We define an instruction sequence $\NZTISp$ that computes $\NZT$ 
according to the same algorithm in less time than $\NZTIS$ as follows:
\begin{ldispl}
\NZTISp \deq 
(\ptst{\tsttt{\Blank}} \conc \fjmp{3} \conc \skptt{1} \conc
 \fjmp{18} \conc {}
\\ \phantom{\NZTISp \deq (}
 \skptt{{-}1} \conc {}
\\ \phantom{\NZTISp \deq (}
 \ntst{\tsttt{0}} \conc \fjmp{3} \conc \settt{\Blank}{{-}1} \conc
 \fjmp{18} \conc {}
\\ \phantom{\NZTISp \deq (}
 \ntst{\tsttt{1}} \conc \fjmp{3} \conc \settt{\Blank}{{-}1} \conc
 \fjmp{3} \conc {}
\\ \phantom{\NZTISp \deq (}
 \settt{0}{0} \conc \halt \conc {}
\\ \phantom{\NZTISp \deq (}
 \ptst{\tsttt{\Blank}} \conc \fjmp{3} \conc \settt{\Blank}{{-}1} \conc
 \fjmp{18} \conc {}
\\ \phantom{\NZTISp \deq (}
 \settt{1}{0} \conc \halt)\rep\;.
\end{ldispl}%
In $\NZTISp$, which is clearly not a single-tape Turing-machine program, 
all instructions of the form $\tsttt{b}$ that are redundant or can be 
made redundant by using instructions of the form $\skptt{d}$ are 
removed.

In~\cite{BM18a}, we have presented instruction sequences that compute 
the restriction of $\NZT$ to $\IAlph^n$, for $n > 0$.
The instruction sequences concerned are instruction sequences that, 
under execution, can act on Boolean registers instead of Turing tapes.

\section{Concluding Remarks}
\label{sect-concl}

We have presented an instantiation of a parameterized algebraic theory 
of single-pass instruction sequences, the behaviours produced by such 
instruction sequences under execution, and the interaction between such 
behaviours and components of an execution environment. 
The parameterized theory concerned is the basis of a line of research in 
which issues relating to a wide variety of subjects from computer 
science have been rigorously investigated thinking in terms of 
instruction sequences (see~\cite{SiteIS}).
In the presented instantiation of this parameterized theory, all 
possible instructions to read out or alter the content of the cell of a 
Turing tape under the tape's head and to optionally move the head in 
either direction by one cell are taken as basic instructions and Turing 
tapes are taken as the components of an execution environment.

The instantiated theory provides a setting  for the development of 
theory in areas such as computability and computational complexity that 
distinguishes itself by offering the possibility of equational reasoning 
and being more general than the setting provided by a known version of 
the Turing-machine model of computation.
Many known and unknown versions of the Turing-machine model of 
computation can be dealt with by imposing apposite restrictions.

We have defined the notion of a single-tape Turing-machine program in 
the setting of the instantiated theory and have provided evidence for 
the claim that the theory provides a suitable setting for the 
development of theory in areas such as computability and computational 
complexity.
Single-tape Turing-machine programs and multiple-tape Turing-machine 
programs make up only small parts of the instruction sequences that can 
be considered in this setting.
This largely explains why it is more general than the setting provided 
by a known version of the Turing-machine model of computation.
From our experience in previous work with a comparable algebraic theory 
of instruction sequences, with instructions that operate on Boolean 
registers instead of Turing tapes, we expect that the generality is 
conducive to the investigation of novel issues in areas such as 
computability and computational complexity.

The given presentation of the instantiated theory is set up in a way 
where the introduction of services, the generic kind of 
execution-environment components from the parameterized theory in 
question, is circumvented.
In~\cite{BM18b}, the presentation of another instantiation of the same 
parameterized theory, with instructions that operate on Boolean 
registers, is set up in the same way.
The distinguishing feature of this way of presenting an instantiation of 
this parameterized theory is that it yields a less involved presentation 
than the way adopted in earlier work based on an instantiation of this 
parameterized theory.

The closed terms of the instantiated theory that are of sort $\InSeq$ 
can be considered to constitute a programming language of which the 
syntax and semantics is defined following an algebraic approach.
This approach is more operational than the usual algebraic approach
which is among other things followed in~\cite{BDMW81a,BWP87a,GM96a}.
A more operational approach is needed to make it possible to investigate 
issues in the area of computational complexity.

Broadly speaking, the work presented in this paper is concerned with 
formalization in subject areas, such as computability and computational
complexity, that traditionally relies on a version of the Turing-machine 
model of computation.
To the best of our knowledge, very little work has been done in this 
area.
Three notable exceptions are~\cite{AR15a,Nor11a,XZU13a}.
However, in those papers, formalization means formalization in a theorem
prover (Matita, HOL4, Isabelle/HOL).
Little or nothing is said in these papers about the syntax and semantics 
of the notations used --- which are probably the ones that have to be 
used in the theorem provers.
This makes it impracticable to compare the work presented in those 
papers with our work, but it is of course clear that the approach 
followed in the work presented in those papers is completely different 
from the algebraic approach followed in our work.

\bibliographystyle{splncs03}
\bibliography{IS}

\end{document}